\documentclass[letterpaper, 10 pt, conference]{ieeeconf}  
\IEEEoverridecommandlockouts                             
\usepackage{cite}
\usepackage{amsmath,amssymb,amsfonts}
\usepackage{algorithmic}
\usepackage{graphicx}
\usepackage{amsmath}
\usepackage{textcomp}
\usepackage{xcolor}
\usepackage{stfloats} 
\usepackage{float}
\usepackage{accents}
\usepackage{tikz}
\usetikzlibrary{calc}
\usepackage{pgfplots}
\usepackage{hyperref}
\usetikzlibrary{shapes, arrows,automata}
\usetikzlibrary{shapes.multipart}
\usepackage[latin1]{inputenc}   
\DeclareMathAlphabet{\pazocal}{OMS}{zplm}{m}{n}
\usepackage{calrsfs}
\usepackage{dsfont}
\let\labelindent\relax
\usepackage{enumitem}
\usepackage{subfigure}
\usepackage{caption}
\usepackage{enumitem}

\usepackage{array}
\newcolumntype{P}[1]{>{\centering\arraybackslash}p{#1}}

\newtheorem{definition}{Definition}
\newtheorem{assumption}{Assumption}
\newtheorem{proposition}{Proposition}
\newtheorem{corollary}{Corollary}
\newtheorem{remark}{Remark}

\newcommand{\figurename}[1]{Fig.{#1}}
\newcommand{\tablename}[1]{Table{#1}}

\tikzstyle{sum} = [draw, fill=blue!20, circle, node distance=1cm]
\tikzstyle{input} = [coordinate]
\tikzstyle{pinstyle} = [pin edge={to-,thin,black}]

\usepackage{cleveref}
\usepackage{tabularx}
\usepackage{longtable} 
\usepackage{colortbl}
\usepackage[thinlines]{easytable}

\usepackage{transparent}

\usepackage{pdfpages} 
\usepackage{afterpage}
\usepackage{lipsum} 
\usepackage{fancyhdr} 
\usepackage{multirow} 
\usepackage{adjustbox} 
\usepackage{amsmath} 
\usepackage{comment} 
\usepackage{color} 
\usepackage{soul} 
\usepackage{eurosym}  
\usepackage{amsfonts}

\usepackage{calrsfs}

\begin{document}
\title{\LARGE{\textbf{Insights into the explainability of\\Lasso-based DeePC for nonlinear systems}}}

\author{Gianluca Giacomelli, Simone Formentin, Victor G. Lopez, Matthias A. M{\"u}ller, and Valentina Breschi
\thanks{This work is partially supported by the FAIR project (NextGenerationEU, PNRR-PE-AI, M4C2, Investment 1.3), the 4DDS project (Italian Ministry of Enterprises and Made in Italy, grant F/310097/01-04/X56), and the PRIN PNRR project P2022NB77E (NextGenerationEU, CUP: D53D23016100001). It is also partly supported by the ENFIELD project (Horizon Europe, grant 101120657).}
	\thanks{Gianluca Giacomelli and Valentina Breschi are with Control Systems Group, Eindhoven University of Technology, 5612AZ Eindhoven, The Netherlands, (e-mails: \textsl{\{g.giacomelli, v.breschi\}@tue.nl}). Simone Formentin is with the Dipartimento di Elettronica, Informazione e Bioingegneria,
		Politecnico di Milano, 20133 Milan, Italy (e-mail: \textsl{simone.formentin@polimi.it}). Victor G. Lopez and Matthias A. M{\"u}ller are with the Institute of Automatic Control, Leibniz University Hannover, 30167 Hannover, Germany (e-mails: \textsl{\{lopez, mueller\}@irt.uni-hannover.de}).}}

\maketitle

\maketitle

\begin{abstract}
    Data-enabled Predictive Control (DeePC) has recently gained the spotlight as an easy-to-use control technique that allows for constraint handling while relying on raw data only. Initially proposed for linear time-invariant systems, several DeePC extensions are now available to cope with nonlinear systems. Nonetheless, these solutions mainly focus on ensuring the controller's effectiveness, overlooking the explainability of the final result. As a step toward explaining the outcome of DeePC for the control of nonlinear systems, in this paper, we focus on analyzing the earliest and simplest DeePC approach proposed to cope with nonlinearities in the controlled system, using a Lasso regularization. Our theoretical analysis highlights that the decisions undertaken by DeePC with Lasso regularization are unexplainable, as control actions are determined by data incoherent with the system's local behavior. This result is true even when the available input/output samples are grouped according to the different operating conditions explored during data collection. Our numerical study confirms these findings, highlighting the benefits of data grouping in terms of performance while showing that explainability remains a challenge in control design via DeePC.
\end{abstract}


\section{Introduction}\label{sec:Intro}
After the seminal works in \cite{coulson2019data,berberich2020data}, one of the main drivers behind the popularity of data-driven predictive control has been its potential to simplify control design, especially for complex systems. Indeed, existing data-driven predictive control strategies directly use raw data to construct the predictor they rely upon, not requiring the designer to undertake any explicit modeling task (which knowingly counts for most of investments in designing advanced controllers~\cite{gevers2005identification,hjalmarsson2005experiment}). While initial efforts to develop data-driven predictive control strategies were mainly directed to linear time-invariant (LTI) systems (see, e.g.,~\cite{dorfler2022bridging, chiuso2025harnessing, verheijen2023handbook}), several extensions of the foundational works in \cite{coulson2019data,berberich2020data} have now been proposed to control nonlinear systems.

Among them, the method proposed in \cite{berberich2022linear} relies on a local linearization of the nonlinear system by iteratively updating the data used to construct the data-driven predictor, with tailored regularization terms exploited to counteract the increasing prediction error as the system moves away from the linearization point. Instead,~\cite{verhoek2023direct} translates general nonlinear dynamics into the linear parameter-varying framework by considering the system's velocity form. On the other hand, when focusing on the class of feedback-linearizable systems,~\cite{alsalti2023data} decomposes the unknown nonlinearities by leveraging an (approximate) basis function. Another approach to design data-driven controllers for nonlinear systems is to exploit the Koopman framework (see, e.g.,~\cite{de2024koopman}), with the resulting approximation leading to a prediction error that increases as the system moves away from its initial condition (see~\cite{iacob2024koopman}). All these approaches have proven to guarantee closed-loop stability as well as performance, especially when the distance between the initial conditions of the estimated linearized system and the nonlinear system is bounded in~\cite{berberich2022linear}. At the same time, the explainability of their outcome, i.e., whether the final user could interpret the decision taken by any of these control schemes, has never been analyzed. Nonetheless, this aspect is crucial for the final users to trust and, ultimately, use these approaches in practice~\cite{voigt2017eu,riva2024towards}. 

\paragraph*{Contribution} To set the cornerstone for the analysis of the explainability of data-driven predictive control schemes for nonlinear systems, we go back to the original work in~\cite{coulson2019data}. Indeed, although conceived for LTI systems, the approach proposed in~\cite{coulson2019data} features a Lasso regularization (for this reason, we denote it as \textit{Lasso-based DeePC}) that might help manage nonlinearities up to a certain distance from the operating condition at which the data used for prediction are collected~\cite{coulson2019regularized}
, shrinking the amount of data actually used for control design.   

To analyze its explainability, we first provide a set of assumptions on the available control data that allow us to provide (for the first time) a possible definition of \textit{explainability} in the predictive control setting driven by data. We then derive the \textit{explicit solution} for the Lasso-based DeePC problem, showing that \textit{the Lasso regularization does not enable such a scheme to prioritize specific subsets of data over others}. This result is key in analyzing the explainability (or lack thereof) of Lasso-based DeePC when the data used to build the predictor are structured in an \textquotedblleft explainable\textquotedblright \ way. Specifically, we introduce a set of data structures that allow one to highlight the different operating conditions under which the data have been collected, thus making the predictor explainable, discussing the role of experiment design in building them. By looking at the explicit solution of Lasso-DeePC when using these explainable matrices, we show once more that no priority can be given through the Lasso regularization alone to specific subsets of data, making the decision of Lasso-DeePC ultimately unexplainable according to our definition of explainability.  

While simplicity is a distinctive trait of Lasso-based DeePC and satisfactory closed-loop performance might be achieved with it under certain conditions even when controlling nonlinear systems as discussed above, our theoretical and numerical analysis thus highlights that understanding its decisions might be difficult.

\paragraph*{Outline} The paper is organized as follows. The setting and goal of this work are introduced in Section~\ref{sec:Setting}, with a focus on the characteristics of the available data as well as our definitions of explainability. The explicit solution for Lasso-based DeePC is presented in Section~\ref{sec:lasso_deepc}, which allows us to discuss its explainability (through the introduction of a set of \textquotedblleft explainable\textquotedblright \ data structures) in Section~\ref{sec:interpretability}. The adherence of our theoretical conclusions to a benchmark numerical case study is discussed in Section~\ref{sec:Numerical}, with the paper ending with some final remarks and directions for future work.

\paragraph*{Notation} Let $\mathbb{N}_{0}$ and $\mathbb{R}_{\geq 0}$ denote the set of natural numbers that include zero and the set of non-negative real numbers, respectively. Given a vector $u \in \mathbb{R}^{n_u}$ and a matrix $A \in \mathbb{R}^{m \times n}$, we define their transposes as $u^{\top}$ and $A^{\top}$, respectively. Moreover, when $A \in \mathbb{R}^{n\times n}$ is invertible, we denote its inverse as $A^{-1}$. When considering a vector  $u \in \mathbb{R}^{n_u}$, inequalities should be seen as component-wise, i.e., $u\geq0$, implies that $u_{j} \geq 0$ for all $j=1,\ldots,n_u$. Zero and identity matrices are respectively denoted as $0$ and $I$, with their dimensions not explicitly reported. A column vector of ones having dimension $n$ is indicated as $\mathds{1}_{n}$. Given a signal $\zeta(t) \in \mathbb{R}^{n_\zeta}$, with $t \in \mathbb{R}_{\geq 0}$ we define its first and second-order derivatives with respect to time as $\dot{\zeta}(t)$ and $\ddot{\zeta}(t)$, respectively. Meanwhile, given a signal $\zeta_{t} \in\mathbb{R}^{n_\zeta}$ for $t \in \mathbb{N}_0$, then ${\zeta_{\left[\tau, \tau'\right]} =\begin{bmatrix} \zeta_{\tau}^\top & \cdots & \zeta_{\tau'}^\top \end{bmatrix}^\top}$, with $\tau,\tau'\in\mathbb{N}_0,~\tau'>\tau$ and the associated Hankel matrix $\pazocal{H}_D(\zeta_{\left[\tau, \tau'\right]})\in\mathbb{R}^{n_\zeta D \times \left( N_\zeta-D+1 \right)}$ of depth $D>0$ is defined as
\begin{equation}\label{eq:Hankel}
    \pazocal{H}_D(\zeta_{\left[\tau, \tau'\right]}) := \begin{bmatrix}
                     \zeta_\tau & \zeta_{\tau+1} & \cdots & \zeta_{\tau'-D}\\
                     \zeta_{\tau+1} & \zeta_{\tau+2} & \cdots & \zeta_{\tau'-D+1}\\
                     \vdots & \vdots & \ddots & \vdots\\
                     \zeta_{\tau+D-1} & \zeta_{\tau+D} & \cdots & \zeta_{\tau'}
\end{bmatrix}.
\end{equation}
Based on this definition, we can introduce that of a persistently exciting signal, formalized in the following. 
\begin{definition}[Persistency of excitation~\cite{willems2005note}]\label{def:persistency_excitation}
    A signal $u_{[0,T-1]}$, with $u_{t} \in \mathbb{R}^{n_u}$ for all $t\in \mathbb{N}_{0}$, is said to be Persistently Exciting (PE) of order $\eta$ if $\mathrm{rank}(\pazocal{H}_{\eta}(u_{[0,T-1]}))=n_u\eta$.  
\end{definition}

\section{Setting \& Goal}\label{sec:Setting}
Consider a nonlinear, time-invariant system $\pazocal{S}$, of which we can only access its inputs $u_t \in \mathbb{R}^{n_u}$ and outputs $y_t \in \mathbb{R}^{n_y}$, for all $t \in \mathbb{N}_{0}$. Our objective is for the system to \textquotedblleft optimally\textquotedblright \ track a predefined input/output reference $(u_t^r,y_t^r)$ spanning $n_p \geq 1$ operating points of the system while satisfying a set of \emph{polyhedral} input $\mathbb{U} \subseteq \mathbb{R}^{n_u}$ and output $\mathbb{Y} \subseteq \mathbb{R}^{n_y}$ constraints.

Suppose that $\pazocal{S}$ is (at least) \emph{locally controllable}~\cite{sussmann1987general} around the $n_p$ operating points spanned by the reference to be tracked. Moreover, suppose that the information available to undertake this (constrained) tracking task is a set of \emph{noiseless}\footnote{For now, we consider an (ideal) noise-free setting to focus on nonlinearities. Nonetheless, we plan to consider noisy data in future works.} input/output data $\pazocal{D}=\{u_{[0,T_d-1]}^{d},y_{[0,T_d-1]}^{d}\}$, collected by \textquotedblleft sufficiently exploring\textquotedblright \ the neighborhood of such $n_p$ operating points. Given the available information, we plan to achieve the desired tasks via DeePC, exploiting Lasso regularization to handle nonlinearities~\cite{coulson2019data}. Specifically, let us store past inputs/outputs collected while closing the loop into
\begin{equation}\label{eq:z_ini} 
z_{\mathrm{ini},t}=\begin{bmatrix}
u_{[t-\rho,t-1]}^{\top} & y_{[t-\rho,t-1]}^{\top}
\end{bmatrix}^{\!\top\!},~~~t \!\in\! \mathbb{N}_{0},
\end{equation}
where $\rho \geq 0$ is the user-defined \textquotedblleft past\textquotedblright \ horizon
. By grouping the available data into the Hankel matrices (see~\eqref{eq:Hankel})
\begin{equation}\label{eq:hankel}
\pazocal{H}_{\rho+L}(u_{[0,T_d-1]}^{d})=\begin{bmatrix}
    U_{P}\\
    U_{F}
\end{bmatrix},~~~\pazocal{H}_{\rho+L}(y_{[0,T_d-1]}^{d})=\begin{bmatrix}
    Y_{P}\\
    Y_{F}
\end{bmatrix},
\end{equation}
where $L\geq 1$ is the prefixed prediction horizon, we can formulate the Lasso-based DeePC problem as:
\begin{subequations}\label{eq:DeePC_lasso}
\begin{align}
&\underset{\substack{g_{t},~u_{f},~y_{f}}}{\mathrm{minimize}}~~J(u_{f},y_{f})+\lambda_g \|g_{t}\|_{1}\\
&\qquad~~\mbox{s.t. }~\begin{bmatrix}
    U_{P}\\
    Y_{P}\\
    U_{F}\\
    Y_{F}
\end{bmatrix}g_{t}=\begin{bmatrix}
z_{\mathrm{ini},t}\\
u_{f}\\
y_{f}
\end{bmatrix},\label{eq:predictive_model}\\
& \qquad \qquad ~~~ u_{k} \in \mathbb{U},~~k \in [0,L-1],\label{eq:constraints1}\\
& \qquad \qquad ~~~ y_{k} \in \mathbb{Y},~~k \in [0,L-1],\label{eq:constraints2}
\end{align}
with $\lambda_g>0$ being a hyperparameter to be tuned, $u_f=u_{[0,L-1]}$, $y_f=y_{[0,L-1]}$, $g_{t} \in \mathbb{R}^{T_d-L-\rho+1}$, and
\begin{equation}\label{eq:cost_DeePC}
J(u_{f},y_{f})\!=\!\!\sum_{k=0}^{L-1} \|y_{k}\!-\!y_{t+k}^{r}\|_{Q\!}^{2}+\|u_{k}\!-\!u_{t+k}^{r}\|_{R}^{2}.
\end{equation}
where $Q \in \mathbb{R}^{n_y \times n_y}$ and $R \in \mathbb{R}^{n_u \times n_u}$ are predefined positive definite weights. This problem is solved at each time instant $t \in \mathbb{N}_{0}$ in a receding horizon fashion, feeding $\pazocal{S}$ with the first optimal input and discarding the rest of the input sequence.
\end{subequations}

Our objective in this work is to analyze whether the Lasso-based DeePC scheme in~\eqref{eq:DeePC_lasso} is \emph{explainable} according to the following definition (adapted from~\cite{Dib2024}). 
\begin{definition}[Explainable Lasso-based DeePC]\label{eq:interpretability_def}
    The DeePC scheme in \eqref{eq:DeePC_lasso} is explainable if its final user can understand how the corresponding optimal input sequence is derived for all possible $z_{\mathrm{ini},t} \in \mathbb{R}^{(n_u+n_y)\rho}$, $\forall t \in \mathbb{N}_{0}$. \hfill $\square$
\end{definition}
We will specialize this broad definition of explainability for the purposes of this paper in the next Subsection (see Definition~\ref{eq:interpretability_loc_def}). 
\begin{remark}[The selector $g_{t}$]
    From here onward, we will often refer to the optimization variable $g_{t} \in \mathbb{R}^{n_g}$, where $n_g=T_d-L-\rho+1$ as the \textquotedblleft selector\textquotedblright. Indeed, also due to the presence of the Lasso regularization \cite{tibshirani1996regression}, such a variable \emph{selects} how much each data column weighs in the design of the optimal control action.
\end{remark}
\subsection{Conditions on the data \& their implications}\label{subsec:setting_data}
Suppose that the data in $\pazocal{D}$ have already been partitioned into $n_p$ subsets $\{\pazocal{D}_{i}\}_{i=1}^{n_p}$, such that
\begin{equation}\label{eq:partitioned_data}
    \bigcup_{i=1}^{n_p}\pazocal{D}_{i}\!=\!\pazocal{D},~~~\pazocal{D}_{i} \!\cap\! \pazocal{D}_{j}\!=\!\emptyset, ~~i\!\neq\! j,~i,j=1,\ldots,n_p,
\end{equation} 
each associated with a neighborhood of one of the operating points spanned by the reference to be tracked. Moreover, suppose that having a \textquotedblleft sufficient exploration\textquotedblright \ of the $n_p$ operating points spanned by the reference implies two conditions on the data comprised in each subset 
\begin{equation}\label{eq:subset}
\pazocal{D}_{i}=\left\{u_{[\tau_{i},\tau_{i}']}^{d},y_{[\tau_{i},\tau_{i}']}^{d}\right\},~~0 \leq \tau_{i}<\tau_i'\leq T_d-1,
\end{equation}
namely
   the inputs $u_{[\tau_{i},\tau_{i}']}$ are PE of order $\rho+L+n$, with $L\geq 1$ being the prefixed prediction horizon of the control scheme to be designed and $\rho\geq 0$, where 
   $n>0$ is 
   the order of $\pazocal{S}$
   . 
   This demand that $\pazocal{D}_i$ should be sufficiently long, i.e., 
\begin{equation}\label{eq:sufficeintly_long}
        T_{d,i}\!=\!\tau_i'\!-\!\tau_i + 1\!\geq\! (n_u\!+\!1)(\rho+\!L+\!n)\!-\!1.
    \end{equation}
Accordingly, we specialize the definition of explainability provided in Definition~\ref{eq:interpretability_def} to the considered setting as follows.
\begin{definition}[Explainability through local behaviors]\label{eq:interpretability_loc_def}
The DeePC scheme in \eqref{eq:DeePC_lasso} is explainable if the corresponding optimal input sequence is generated by using only data from $\pazocal{D}_{i}$ (see~\eqref{eq:subset}) when the initial condition $z_{\mathrm{ini}}$ and the reference to be tracked are in a neighborhood of the $i$-th operating point of $\pazocal{S}$, for all $i=1,\ldots,n_p$. \hfill $\square$
\end{definition}
This definition implies that whenever $z_{\mathrm{ini}}$ and $\{(u_{t+k}^{r},y_{t+k}^{r})\}_{k=0}^{L-1}$ is around a single operating point, $g_t$ should be able to select only the data associated with the \textquotedblleft right\textquotedblright \ (desired) operating condition. Meanwhile, if $z_{\mathrm{ini}}$ or $\{(u_{t+k}^{r},y_{t+k}^{r})\}_{k=0}^{L-1}$ spans multiple operating points on the prediction horizon, then $g_t$ should be able to select the data associated with the \textquotedblleft right\textquotedblright \ operating mode at each prediction step. While it is not structurally possible to have different selectors for different time instants\footnote{This structural change is the object of our ongoing work.}, 
explainability according to Definition~\ref{eq:interpretability_loc_def} could still be achieved if the data associated to the \textquotedblleft right\textquotedblright \ operating modes would be weighted more (through $g_t$) than the others. 

Under these conditions, the objective of this paper is ultimately to answer the following question: \emph{Is Lasso-based DeePC explainable according to Definition~\ref{eq:interpretability_loc_def}?}

\section{Lasso-based DeePC: an explicit solution}\label{sec:lasso_deepc}
A 
first step toward analyzing the explainability of Lasso-based DeePC in the sense of Definition~\ref{eq:interpretability_loc_def}, could be the derivation of its explicit solution. However, due to the features of the Hankel data matrices (that are only full row rank)  we need an additional $L^2$ regularization on $g_t$ to guarantee the uniqueness of the solution of the DeePC problem (see \cite[Lemma 1]{Breschi2023}). We thus consider the modified Lasso-based DeePC
\begin{subequations}\label{eq:DeePC_lasso3}
\begin{align}
&\underset{\substack{g_{t},~u_{f},~y_{f}}}{\mathrm{minimize}}~~J(u_{f},y_{f})+\lambda_g \|g_t\|_{1}+\lambda_{2}\|g_t\|_2^2\label{eq:Cost_DeePC3}\\
&\qquad~~\mbox{s.t. }~\begin{bmatrix}
    U_{P}\\
    Y_{P}\\
    U_{F}\\
    Y_{F}
\end{bmatrix}g_{t}=\begin{bmatrix}
z_{\mathrm{ini},t}\\
u_{f}\\
y_{f}
\end{bmatrix},\\
& \qquad \qquad ~~~ u_{k} \in \mathbb{U},~~k \in [0,L-1],\label{eq:constraints1_2}\\
& \qquad \qquad ~~~ y_{k} \in \mathbb{Y},~~k \in [0,L-1],\label{eq:constraints2_2}
\end{align}
\end{subequations}
where $\lambda_2>0$ is here assumed to be infinitesimally small.

We then derive the explicit solution of \eqref{eq:DeePC_lasso3}. To this end, let us introduce $\boldsymbol{y}^{r}=y_{[t,t+L-1]}^{r}$ and $\boldsymbol{u}^{r}=u_{[t,t+L-1]}^{r}$, simplify the notation by dropping the dependence on $t$ in \eqref{eq:DeePC_lasso} and defining
\begin{equation}
    Z_P=\begin{bmatrix}
        U_P\\
        Y_P
    \end{bmatrix},
\end{equation}
so that the predictor in \eqref{eq:predictive_model} becomes
\begin{equation}
    \begin{bmatrix}
        Z_P\\
        U_F\\
        Y_F
    \end{bmatrix}g=\begin{bmatrix}
        z_{\mathrm{ini}}\\
        u_{f}\\
        y_{f}
    \end{bmatrix}.
\end{equation}
Moreover, let us rewrite the polyhedral input and output constraints explicitly as inequalities on $g$ (see~\cite{Breschi2023}), i.e.,
\begin{equation}\label{eq:polyhedral_constraints}
    \pazocal{G}_{F}g\leq \gamma,
\end{equation}
with $\pazocal{G}_{F}\in\mathbb{R}^{n_c\times n_g}$ and $\gamma\in\mathbb{R}$ describing the $n_c$ polyhedral constraints on the selector stemming from \eqref{eq:constraints1}-\eqref{eq:constraints2}. Following the same steps of \cite{he2011lasso}, let us rewrite the optimization variable $g$ as
\begin{subequations}
    \begin{equation}\label{eq:decomposition}
        g=g^{+}-g^{-},
    \end{equation}
such that $g^{+},g^{-} \in \mathbb{R}^{n_g}$ are both non-negative, i.e., $g_{k}^{+}\geq 0$ and $g_{k}^{-}\geq 0$ for all $k=1,\ldots,n_g$, and the cost in \eqref{eq:cost_DeePC} can be rewritten as
\begin{equation}\label{eq:new_cost_DeePC}
    J(g^{+},g^{-})\!=\!\|Y_F(g^{+\!}\!-\!g^{-})\!-\!\boldsymbol{y}^{r}\|_{\pazocal{Q}}^{2}+\|U_F(g^{+\!}\!-\!g^{-})\!-\!\boldsymbol{u}^{r}\|_{\pazocal{R}}^{2},
\end{equation}
where $\pazocal{Q}=\mathrm{diag}(Q,\ldots,Q)$ and $\pazocal{R}=\mathrm{diag}(R,\ldots,R)$.
\end{subequations}
According to the previous definition and our assumption on the input/output constraints, \eqref{eq:DeePC_lasso} thus becomes
\begin{subequations}\label{eq:DeePC_lasso2}
\begin{align}
&\underset{g^{+\!},~g^{-}}{\mathrm{minimize}}~J(g^{+\!},g^{-})\!+\!\lambda_g\mathds{1}_{n_g}^{\top\!}(g^{+\!}+\!g^{-})\!+\!\lambda_{2}\| g^{+\!}\!-\!g^{-}\|_2^2\\
&\qquad~~\mbox{s.t. }~
    Z_{P}(g^{+}-g^{-})=
z_{\mathrm{ini}},\label{eq:predictive_model2}\\
& \qquad \qquad~~ \pazocal{G}_{F}(g^{+}-g^{-})\leq \gamma,\\
& \qquad \qquad~~ g_{k}^{+} \geq 0,~~g_{k}^{-}\geq 0,~~k\!=\!1,\ldots,n_g.
\end{align}
\end{subequations}
Relying on this reformulation, we can now consider its associated Lagrangian, namely 
\begin{align}\label{eq:Lagrangian}
    \nonumber &\mathcal{L}(g^{+},g^{-\!},\beta,\mu,\mu^{+},\mu^{-})=J(g^{+},g^{-})+\lambda_g\mathds{1}_{n_g}^{\top}(g^{+\!}+g^{-})\\
    \nonumber &\quad \quad +\lambda_{2}\|g^{+}-g^{-}\|_{2}^{2}+\beta^{\top}[Z_{P}(g^{+\!}\!-\!g^{-})\!-\!z_{\mathrm{ini}}]\!\\
    &\quad\quad +\mu^{\!\top}\!\left[\pazocal{G}_{F}(g^{+\!}\!-\!g^{-})-\!\gamma\right]-\!(\mu^{+})^{\!\top}g^{+\!}-\!(\mu^{-})^{\!\top}g^{-},
\end{align}
where $\beta \in \mathbb{R}^{\rho(n_u+n_y)}$, 
$\mu \in \mathbb{R}^{n_c}$, $\mu^{+}\in\mathbb{R}^{n_g}$ and $\mu^{-}\in\mathbb{R}^{n_g}$ are the Lagrangian multipliers associated with the equality and inequality constraints in~\eqref{eq:DeePC_lasso2}, respectively. 
We can then  
write the Karush-Kuhn-Tucker (KKT) conditions for each component of the optimization variables $g^{+}$ and $g^{-}$ in \eqref{eq:DeePC_lasso2} as follows:
\begin{subequations}\label{eq:KKTs}  
\begin{align}
    \nonumber & 2Y_{F}^{\top\!}\pazocal{Q}\varepsilon^{y}(g^{\star})+2U_{F}^{\top}\pazocal{R}\varepsilon^{u}(g^{\star})+\lambda_g\mathds{1}+\!\\
    &\qquad\qquad ~~~+2\lambda_2g^{\star}+Z_{P}^{\top}\beta+\pazocal{G}_{F}^{\top}\mu-\!\mu^{+\!}\!=\!0,\label{eq:KKT1}\\
    \nonumber & \!-\!2Y_{F}^{\top\!}\pazocal{Q}\varepsilon^{y}(g^{\star})-2U_{F}^{\top}\pazocal{R}\varepsilon^{u}(g^{\star})\!+\!\lambda_g\mathds{1}+\\
    &\qquad\qquad ~~~~-2\lambda_2g^{\star}-Z_{P}^{\top}\beta-\pazocal{G}_{F}^{\top}\mu-\mu^{-\!\!}\!=\!0,\label{eq:KKT2}\\
    & g_{k}^{\star,+} \geq 0, ~~ g_{k}^{\star,-} \geq 0,~~k=1,\ldots,n_g, \label{eq:KKT3}\\
     & \mu_{k}^{+} \geq 0, ~~ \mu_{k}^{-} \geq 0,~~k=1,\ldots,n_g, \label{eq:KKT4}\\
    & \mu_{k}^{+}g_{k}^{\star,+}=0,~~\mu_{k}^{-}g_{k}^{\star,-}=0,~~k=1,\ldots,n_g, \label{eq:KKT5}\\
    & Z_{P} g^{\star}-z_{\mathrm{ini}}=0,\label{eq:KKT6}\\
    & \pazocal{G}_{F}g^{\star} \leq \gamma,\label{eq:KKT7}\\
    & \mu^{\top}[\pazocal{G}_{F}g^{\star}-\gamma]=0,\label{eq:KKT8}\\
    & \mu \geq 0,\label{eq:KKT9}
    \end{align}
\end{subequations}
where $\varepsilon^{y}(g^{\star})=Y_{F}g^{\star}\!\!-\boldsymbol{y}^{r}$ and $\varepsilon^{u}(g^{\star})=U_{F}g^{\star}\!\!-\boldsymbol{u}^{r}$, we do not explicitly mention the dimension of $\mathds{1}$ (see \eqref{eq:Lagrangian}) and we have exploited \eqref{eq:decomposition} to simplify the notation.

By following the same steps of \cite{Breschi2023}, let us consider a generic set of active constraints and split $\mu$ as $\mu=\begin{bmatrix}
    \tilde{\mu}^{\top} \!&\! \hat{\mu}^{\top}
\end{bmatrix}^\top$ such that
\begin{subequations}\label{eq:constraints_partition}
    \begin{align}
        & \hat{\pazocal{G}}_{F}g^{\star}-\hat{\gamma}<0,~~\mbox{ and }~~\hat{\mu}\!=\!0,\label{eq:inactive}\\
        & \tilde{\pazocal{G}}_{F}g^{\star}-\tilde{\gamma}=0,~~\mbox{ and }~~\tilde{\mu}\!>\!0,\label{eq:active}
    \end{align}
\end{subequations}
by complementarity slackness (see~\eqref{eq:KKT8}-\eqref{eq:KKT9}). While the inequalities in \eqref{eq:inactive} dictate the polyhedral region where the solution we derive next is valid, \eqref{eq:active} can be combined with \eqref{eq:KKT6} to obtain
\begin{equation}\label{eq:equality constraints}
    \underbrace{\begin{bmatrix}
        Z_{P}\\
        \tilde{\pazocal{G}}_{F}
    \end{bmatrix}}_{=\tilde{G}}g^{\star}-\underbrace{\begin{bmatrix}
        z_{\mathrm{ini}}\\
        \tilde{\gamma}
    \end{bmatrix}}_{=\tilde{b}}=0.
\end{equation}
Let us now make the following assumption on $\tilde{G}$ (see \cite[Assumption 2]{Breschi2023}).
\begin{assumption}\label{assump:indep}
    The rows of $\tilde{G}$ in \eqref{eq:equality constraints} are linearly independent. 
\end{assumption}
Summing and subtracting \eqref{eq:KKT1} and \eqref{eq:KKT2} we respectively get
\begin{subequations}
    \begin{align}
        & 2\lambda_g\mathds{1}=\mu^{+}+\mu^{-},\label{eq:KKT1bis}\\
        \nonumber & 2Y_{F}^{\top\!}\pazocal{Q}\varepsilon^{y}(g^{\star})\!+\!2U_{F}^{\top}\pazocal{R}\varepsilon^{u}(g^{\star})\!+\!4\lambda_2g^{\star}+\\
        &\qquad \qquad \qquad \qquad +\tilde{G}^\top\delta\!-\!\frac{1}{2}(\mu^{+\!\!}-\!\mu^{-})\!=\!0,\label{eq:KKT2bis}
    \end{align}
\end{subequations}
where $\delta=\begin{bmatrix}
\beta^{\top} \!\!&\!\! \tilde{\mu}^{\!\top}
\end{bmatrix}^{\top}$. By defining $W_F=2Y_{F}^{\top\!}\pazocal{Q}Y_F+2U_{F}^{\top}\pazocal{R}U_F+4\lambda_2 I$ (that is positive definite and, thus, invertible by definition) and $c_{F}=2Y_{F}^{\top\!}\pazocal{Q}\boldsymbol{y}^{r}+2U_{F}^{\top}\pazocal{R}\boldsymbol{u}^{r}$, manipulating the latter equation we get:
\begin{equation}\label{eq:g_partial}
    g^{\star}=-W_{F}^{-1}\left[\tilde{G}^\top\delta-\frac{1}{2}(\mu^{+\!\!}-\!\mu^{-})-c_F\right].
\end{equation}
By replacing this into \eqref{eq:equality constraints}, we can thus find the explicit form of $\delta$, i.e.,
\begin{equation}\label{eq:delta}
    \delta \!=\! (\tilde{G}W_F^{-1}\tilde{G}^{\top})^{-1\!\!}\left[\frac{1}{2}\tilde{G}W_F^{-1}(\mu^{+\!\!\!}-\!\mu^{-})\!+\!\tilde{G}W_F^{-1}c_{F}\!-\!\tilde{b}\right]\!,
\end{equation}
which, when replaced in \eqref{eq:g_partial} leads to an expression of $g^{\star}$ that depends only on data and on the Lagrange multipliers $\mu^{+}$ and $\mu^{-}$. By following the same reasoning of \cite{he2011lasso}, let us now distinguish the following cases:
\begin{enumerate}
    \item for all $k \in \{1,\ldots,n_g\}$ such that $g_{k}^{\star,+}>0$, 
    we have $\mu_{k}^{+}=g_{k}^{\star,-}=0$ and $\mu_{k}^{-}=2\lambda_g$;
     \item for all $k \in \{1,\ldots,n_g\}$ such that $g_{k}^{\star,-}>0$, 
     it holds $\mu_{k}^{-}=g_{k}^{\star,+}=0$ and $\mu_{k}^{+}=2\lambda_g$;
     \item for all $k \in \{1,\ldots,n_g\}$ such that $g_{k}^{\star,+}=g_{k}^{\star,-}=0$, 
     we have $\mu_{k}^{+}+\mu_{k}^{-}=2\lambda_g$,
\end{enumerate}
which stem from \eqref{eq:KKT3},~\eqref{eq:KKT4}, \eqref{eq:KKT5}, and \eqref{eq:KKT1bis}. Note that $g_{k}^{\star,+}>0$ and $g_{k}^{\star,-}>0$ cannot be verified simultaneously for~\eqref{eq:KKT1bis} as $\lambda_g>0$. 

Accordingly, the following holds
\begin{equation}\label{eq:multipliers_values}
    \mu_{k}^{+}-\mu_{k}^{-}=\begin{cases}
        -2\lambda_g, &\mbox{ if } g_{k}^{\star,+}>0,\\
        2\lambda_g, &\mbox{ if } g_{k}^{\star,-}>0,\\
        c_{k}^{\mu}(z_{\mathrm{ini}}) &\mbox{ if } g_{k}^{\star,+\!\!}=g_{k}^{\star,-\!\!}=0,
    \end{cases}
\end{equation}
where $c_{k}^{\mu}(z_{\mathrm{ini}})$ is a constant depending on the data\footnote{An explicit expression for this constant can be obtained by replacing \eqref{eq:delta} into \eqref{eq:g_partial}, considering the $k$-th component of the solution and imposing $g^{\star}=0$ and $\mu_{k}^{+}+\mu_{k}^{-}=2\lambda_g$.}, the initial conditions $z_{\mathrm{ini}}$ of the 
DeePC problem, as well as the features of its cost in \eqref{eq:Cost_DeePC3}. 
This result allows us to show that when \eqref{eq:inactive} and \eqref{eq:active} hold, then $g^{\star}=\tilde{\pazocal{F}}z_{\mathrm{ini}}+\tilde{f}$, where both $\tilde{\pazocal{F}}$ and $\tilde{f}$ are 
determined by the available data and the features of the DeePC problem.

It ultimately follows that the explicit solution of (modified) Lasso-based DeePC in \eqref{eq:DeePC_lasso3} is a Piecewise Affine (PWA) law, i.e.,
\begin{equation}\label{eq:PWAlaw}
    g^{\star}=\begin{cases}
    \pazocal{F}_{1}z_{\mathrm{ini}}+f_{1},~&\mbox{ if }~\pazocal{P}_{1}z_{\mathrm{ini}}\leq p_{1},\\
    \vdots \\
    \pazocal{F}_{M}z_{\mathrm{ini}}+f_{M},~&\mbox{ if }~\pazocal{P}_{M}z_{\mathrm{ini}}\leq p_{M},
    \end{cases}
\end{equation}
where $\pazocal{P}_{m}\in\mathbb{R}^{n_m\times\rho(n_u+n_y)}$ and $p_{m}\in\mathbb{R}^{n_m}$, $m=1,\dots,M$, $M\geq 1$ are dictated by the number of possible combinations of active constraints, while the fact that the partition is driven by $z_{\mathrm{ini}}$ can be shown by merging \eqref{eq:g_partial}, \eqref{eq:delta} and \eqref{eq:multipliers_values} and replacing the resulting expression for $g^{\star}$ into \eqref{eq:inactive}. 
\begin{proposition}[Explicit solution
]
\label{prop:explicit_lasso} Let $\lambda_g>0$ and $\lambda_2>0$, with $\lambda_2$ infinitesimally small. Under Assumption~\ref{assump:indep}, the explicit solution of \eqref{eq:DeePC_lasso3} is a PWA law in the initial conditions $z_{\mathrm{ini}}$ (see \eqref{eq:PWAlaw}), where all the available data, the references and the weights characterizing the cost function in \eqref{eq:cost_DeePC}  and the regularization parameters $\lambda_g$ and 
$\lambda_2$
drive the gain of the local affine control laws and its partitions.   
\end{proposition}
\begin{proof}
The proof follows from \eqref{eq:decomposition}-\eqref{eq:multipliers_values}.
\end{proof}
By relying on the expressions in \eqref{eq:g_partial}, \eqref{eq:delta} and \eqref{eq:multipliers_values}, it can be shown that the value of each component of the optimal selector resulting from the considered modified Lasso-based DeePC is driven by $(i)$ the magnitude of \emph{all} data, $(ii)$ the weights and references in \eqref{eq:cost_DeePC}, as well as the features of the polyhedral constraints in \eqref{eq:constraints1}-\eqref{eq:constraints2}, and $(iii)$ the 
hyperparameters $\lambda_g$ and $\lambda_2$, with the first solely magnifying the quantities above and the second (as it is assumed to be small) only guaranteeing uniqueness while not considerably affecting the solution. 
To showcase this, let us consider the case in which none of the components of $g^{\star}$ is driven to zero for the considered choice of $\lambda_g$. In this case, the optimal solution for the fixed set of active and inactive constraints dictated by \eqref{eq:constraints_partition} is given by
\begin{equation}
    g^{\star\!}\!=\!-\lambda_g W_F^{-1}\!\left\{\tilde{G}^{\top}(\tilde{G}W_F^{-1}\tilde{G}^{\top})^{-1}\tilde{G}W_F^{-1\!}-\!I\right\}e+c_{e},
\end{equation}
with $c_e\!=\!W_{F}^{-1\!}\left\{\!\tilde{G}^{\top\!}(\tilde{G}W_F^{-1}\tilde{G}^{\top})^{-1}(\tilde{G}W_F^{-1}c_{F}\!-\!\tilde{b})\!-\!c_{F}\!\right\}$ and $e_{k}=-1$ if $g_{k}^{+}>0$ or it is $1$ if $g_{k}^{-}>0$, for $k=1,\ldots,n_g$. The features of the explicit solution of the modified Lasso-based DeePC in \eqref{eq:DeePC_lasso3} ultimately indicate that there is no way to prioritize the use of subsets of data (e.g., weighting them more than others) if not because of differences in their magnitudes. This result will be ultimately key to discussing the explainability (or lack thereof) of Lasso-based DeePC according to Definition~\ref{eq:interpretability_loc_def}.
\section{Explainability \& Lasso-DeePC}\label{sec:interpretability}
While the general formulation of Lasso-DeePC in \eqref{eq:DeePC_lasso} does not exploit (nor explicitly account for) the grouping of data assumed in our setting (Section~\ref{subsec:setting_data}), we now explicitly consider it for the analysis of the explainability of \eqref{eq:DeePC_lasso} according to Definition~\ref{eq:interpretability_loc_def}. To this end, we first discuss how operating conditions can be highlighted either through experiment design or clustering the collected data, using these matrices to reformulate \eqref{eq:DeePC_lasso} and analyzing the explainability of the resulting (explicit) solution.
\subsection{Strategies to highlight operating conditions in data}
Building a predictor that is explainable in the sense of Definition~\ref{eq:interpretability_loc_def}, i.e., a predictor where local behaviors are distinguished and distinguishable by the user, is key to explaining the solution of Lasso-based DeePC. Several data structures can be considered to achieve this result, depending on the possibility of designing tailored \emph{data collections} to accomplish this goal.
\paragraph*{Mosaic matrices} A mosaic matrix is a data structure that can be used to build the data-driven predictor in \eqref{eq:predictive_model} while naturally mirroring local behaviors. Indeed, grouping the data associated with the different subsets $\pazocal{D}_{i}$, with $i=1,\ldots,n_p$ (see \eqref{eq:subset}), we can define the predictor based on the mosaic matrix as
\begin{equation}\label{eq:mosaic_predictor}
    \begin{bmatrix}
        \pazocal{M}_{\rho+L}(u_{[0,T_d-1]}^{d})\\
        \pazocal{M}_{\rho+L}(u_{[0,T_d-1]}^{d})
    \end{bmatrix}g_{t}=\begin{bmatrix}
        z_{\mathrm{ini},t}\\
        u_{f}\\
        y_f
    \end{bmatrix}
\end{equation}
where
\begin{subequations}
\begin{align}
    &  \pazocal{M}_{\rho+L}(u_{[0,T_{d,1}-1]}^{d})\!=\!\!\begin{bmatrix}
        \pazocal{H}_{\rho+L}(u_{[\tau_{1},\tau_{1}']}^{d}) \!\!\!\!\!\!& \cdots &\!\!\!\!\! \pazocal{H}_{\rho+L}(u_{[\tau_{n_p},\tau_{n_p}']}^{d})
    \end{bmatrix}\!,\\
    &  \pazocal{M}_{\rho+L}(y_{[0,T_d-1]}^{d})\!=\!\!\begin{bmatrix}
        \pazocal{H}_{\rho+L}(y_{[\tau_{1},\tau_{1}']}^{d}) \!\!\!\!\!\!& \cdots &\!\!\!\!\!\pazocal{H}_{\rho+L}(y_{[\tau_{n_p},\tau_{n_p}']}^{d})
    \end{bmatrix}\!.
\end{align}
\end{subequations}
Note that, in this case, the dimension of the selector becomes $n_g=T_d-n_p(L+\rho-1)$. 

These matrices can be readily constructed if separate experiments are carried out by exciting the system around each operating condition, naturally providing a separation between data associated with different operating modes, as in \cite{verdult2004subspace} and \cite[Section E.1]{faye2024computationally}. Alternatively, clustering techniques (e.g., k-means \cite{hartigan1979k}) can be used to assign each data point in $\pazocal{D}$ to an operating mode and, then, construct the Hankel matrix.  
\paragraph*{Explainable Hankel matrices} Whenever tailored experiments can be performed to cyclically make the system operate around each operating condition of interest for enough steps (see \eqref{eq:sufficeintly_long}), one obtains a trajectory with (uninterrupted) subsequences associated with the different $n_p$ operating points (OPs) of the system. This choice in experiment design naturally leads to explainable Hankel matrices, i.e., where one can isolate blocks of data associated with a single operating condition. Note that some blocks in the explainable Hankel matrices will nonetheless contain data collected when transitioning from one operating condition to the other, \textquotedblleft explaining\textquotedblright \ also these intermediate behaviors.    
\paragraph*{Explainable Page matrices} Similarly to explainable Hankel matrices, explainable Page matrices can also be considered if one can perform an experiment exciting the system around each operating mode in a cyclic way. As the entries in each column of a Page matrix are not repeated (see, e.g.,~\cite{iannelli2021design}), longer experiments are nonetheless needed to construct this kind of (explainable) data matrix.
\subsection{Lasso-DeePC with grouped matrices \& its explainability}
Let us assume that the data have been grouped into a Mosaic matrix, an explainable Hankel, or an explainable Page matrix. Consequently, the predictor in \eqref{eq:predictive_model} can be recast as
\begin{equation}\label{eq:predictor_grouped}
    \sum_{j=1}^{\nu}\begin{bmatrix}
        \pazocal{H}_{\rho+L}(u_{[\tau_{j},\tau_{j}']}^{d})\\
        \pazocal{H}_{\rho+L}(y_{[\tau_{j},\tau_{j}']}^{d})
    \end{bmatrix}g_{t}^{j}=\sum_{j=1}^{\nu}\begin{bmatrix}
        Z_P^j\\
        U_F^j\\
        Y_F^j
    \end{bmatrix}g^j_{t}=\begin{bmatrix}
        z_{\mathrm{ini},t}\\
        u_f\\
        y_f
    \end{bmatrix}\!,
\end{equation}
where $\nu=n_{p}$ if the data are organized into a Mosaic matrix while it corresponds to the number of sub-Hankels/sub-Pages associated with the different modes (or the transitions among them) that are induced by experiment design when considering an explainable Hankel/Page matrices. Accordingly, $g^j \in \mathbb{R}^{n_{j}}$, for $j=1,\ldots,\nu$, are the components of the decision variable associated with the $j$-th group in the data, with $\sum_{j=1}^{\nu}n_j=n_g$. 

Despite the grouping, we can still follow the same steps carried out in Section~\ref{sec:lasso_deepc} to obtain an explicit solution of the modified Lasso-DeePC in \eqref{eq:DeePC_lasso3} with the predictor in~\eqref{eq:predictor_grouped}. This procedure\footnote{We do not report the explicit steps, since they exactly unfold as those reported in Section~\ref{sec:lasso_deepc}.} allows us to prove that 
\begin{equation}\label{eq:explicit_grouped}
    g^{j,\star}=\begin{cases}
    \pazocal{F}_{1}^{j}z_{\mathrm{ini}}+f_{1}^{j},~&\mbox{ if }~\pazocal{P}_{1}^{j}z_{\mathrm{ini}}\leq p_{1}^{j},\\
    \vdots \\
    \pazocal{F}_{M}^{j}z_{\mathrm{ini}}+f_{M}^{j},~&\mbox{ if }~\pazocal{P}_{M}^{j}z_{\mathrm{ini}}\leq p_{M}^{j},
    \end{cases}
\end{equation}
where $\pazocal{F}_{m}^{j}$, $f_{m}^{j}$, $\pazocal{P}_{m}^{j}$ and $p_{m}^{j}$ depend on all data and the features of the 
DeePC problem (i.e., the polyhedral constraints and the cost and regularization weights) irrespective of the grouping embodied by the adopted explainable data structure, for all $j=1,\ldots,\nu$ and $m=1,\ldots,M$. This result is formalized in the following corollary.
\begin{corollary}[Explicit solution \& grouped matrices] 
Let $\lambda_g>0$ and $\lambda_2>0$, with $\lambda_2$ be infinitesimally small. Under Assumption~\ref{assump:indep}, the explicit solution of \eqref{eq:DeePC_lasso3} with grouped data is a PWA law in the initial conditions $z_{\mathrm{ini}}$ (see \eqref{eq:explicit_grouped}), where all the available data drive the gain of the local affine control laws and its partition, the references and the weights characterizing the cost function in \eqref{eq:cost_DeePC} and the regularization parameters.
\end{corollary}
\begin{proof}
    The proof follows by performing the same steps carried out from \eqref{eq:decomposition} to \eqref{eq:multipliers_values}, yet explicitly considering the grouping of the predictor.
\end{proof}
Despite using data matrices that are grouped in an explainable way (according to Definition~\ref{eq:interpretability_loc_def}), this result indicates (as one could already have expected by looking at the explicit solution without grouped data) that the modified Lasso-DeePC is not explainable according to Definition~\ref{eq:interpretability_loc_def}. Indeed, this control strategy does not discriminate between data belonging to different operating regimes of the system if not numerically, i.e., because of the different magnitudes of the data, even if they are organized to reflect differences in the system's operating condition.
\section{A benchmark case study: the unbalanced disk}\label{sec:Numerical}
To empirically validate the conclusions drawn on Lasso-DeePC, we consider the DC motor connected to an unbalanced disk, conventionally used as a benchmark in the linear-parameter varying setting (see, e.g.,\cite{den2021lpvcore}). This system is described by the following difference equation
\begin{equation}\label{eq:pendulum}
\begin{aligned}
    \ddot{y}(t)=\alpha_{1}\cos{(y(t))}+\alpha_{2}\dot{y}(t)+\alpha_{3}u(t),
\end{aligned}
\end{equation}
where $y(t)$~[rad] is the angular position of the unbalanced disk, $u(t)$~[V] is the input voltage, and the coefficients $\alpha_{1}=127.37$~[rad~s$^{-2}$], $\alpha_2=-2.50$~[s$^{-1}$] and $\alpha_3=26.25$~[rad~V$^{-1}$~s$^{-2}$] are dictated by the physical characteristics of the system (see \cite{GitGerben}).

For a sampling time of $0.01$~[s], the Lasso-DeePC scheme has been implemented considering $Q=100$ and $R=1$ in \eqref{eq:cost_DeePC}, setting the prediction horizon to $L=30$ and the past horizon $\rho=40$. Meanwhile, values of $\lambda_g \in [10^{-5},10^{5}]$ for each order of magnitude in this range have been tested. Due to the voltage limits on the DC motor, the input at each time instant has been constrained to $\mathbb{U}=[-10,10]$~[V], while no output constraint has been considered. The output reference to be tracked has been selected so that the system transitions from two distinct input/output operating points\footnote{The input reference is obtained as the one making the output reference an equilibrium for the system.} (OPs), fixing $z_{\mathrm{ini,0}}$ to the values of the initial input/output references. The former choice allows us to $(i)$ highlight the importance of data collection, $(ii)$ showcase the possible benefits of using grouped data structures, as well as $(iii)$ discuss the explainability of Lasso-DeePC\footnote{The code to reproduce our results is available at \url{https://github.com/GiacomelliGianluca/Explainability_Lasso_DeePC}}. 

\subsection{The importance of data collection}
\begin{figure}[!tb]
\centering
\begin{tabular}{cc}
\subfigure[Single OP Hankel\label{fig:single}]{\includegraphics[scale=.5,clip,trim=0cm 0cm 0cm .5cm]{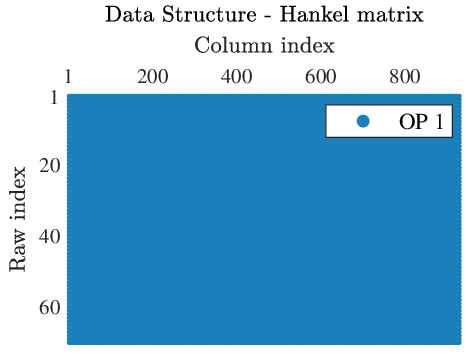}} & \subfigure[Two OPs Hankel\label{fig:2OP_H}]{\includegraphics[scale=.5,clip,trim=0cm 0cm 0cm .5cm]{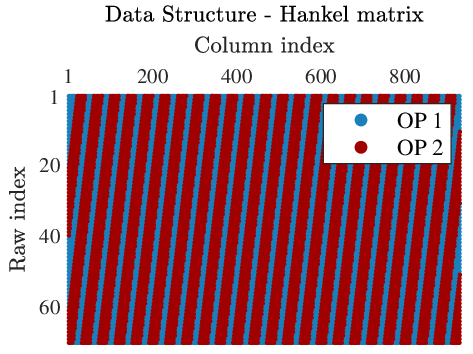}}\\
\multicolumn{2}{c}{\subfigure[Two-blocks Hankel\label{fig:gH}]{\includegraphics[scale=.5,clip,trim=0cm 0cm 0cm .5cm]{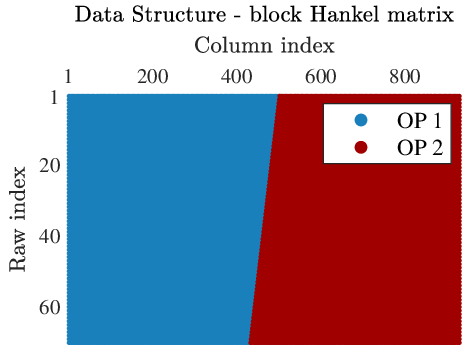}}}
\end{tabular}
\caption{Considered data structures: Hankel matrices comprising data collected (a) around a single operating point, (b) when \textquotedblleft rapidly\textquotedblright \ transitioning between two operating points multiple times, (c) transitioning from one operating point to the other only once.}
    \label{fig:data_structures}
\end{figure}
Before discussing the performance of Lasso-DeePC, let us introduce the data and associated structures used for predictive control design. Specifically, we collect $T_d=1000$ data points by performing three kinds of experiments. First, we only excite the system for it to operate around $0.17$~[rad], thus leading to a dataset exploring a single operating point (see \figurename{~\ref{fig:single}}). We then excite the system to transition between one operating point (i.e., $0.17$~[rad]) and another one (either $1.57$~[rad] or $2.97$~[rad], depending on the reference to be tracked\footnote{The considered OPs correspond to the constant set points for the angular position. 
}) multiple times, as shown in \figurename{~\ref{fig:2OP_H}}. Lastly, we consider a dataset generated by exciting the system for it to operate first around $0.17$~[rad] and then around either $1.57$~[rad] or $2.97$~[rad], with a single transition between the two operating points. Note that the data structure associated with this last collection procedure (see \figurename{~\ref{fig:gH}}) is an explainable Hankel matrix, which resembles a Mosaic matrix, and we refer to it as \textquotedblleft block Hankel\textquotedblright. Nonetheless, it features a subset of columns where data span multiple OPs, ultimately carrying information on the transition between one operating point of interest and the other.  

Irrespective of the choice of $\lambda_g$ among its considered values, Lasso-DeePC formulated with data collected around a single operating point is never feasible, when tracking a reference switching from $0.17$~[rad] to either $1.57$~[rad] or $2.97$~[rad]. This result is somehow expected, as the data do not satisfy the conditions highlighted in Section~\ref{subsec:setting_data}, thus not being informative to characterize the system in the different OPs of interest. On the other hand, both the other data structures lead to feasible solutions to Lasso-DeePC, highlighting the importance of exploring all the operating points one is interested to track in closed-loop during data collection.

\begin{remark}
    Since Page matrices are known to allow noise handling better than Hankel matrices but they require more data to be constructed, we do not consider them in this preliminary work. We thus defer the experimental evaluation of the benefits of explainable Page matrices to future work.
\end{remark}
\subsection{The impact of $\lambda_g$}
\begin{figure}[!tb]
\centering
\begin{tabular}{cc}
   \subfigure[Case 1]{\includegraphics[scale=.5,trim=0cm 0cm 0cm 0.45cm,clip]{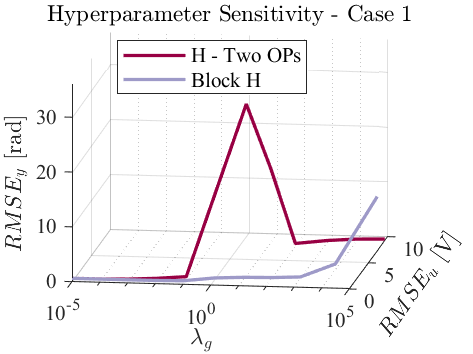}}  & \subfigure[Case 2]{\includegraphics[scale=.5,trim=0cm 0cm 0cm 0.45cm,clip]{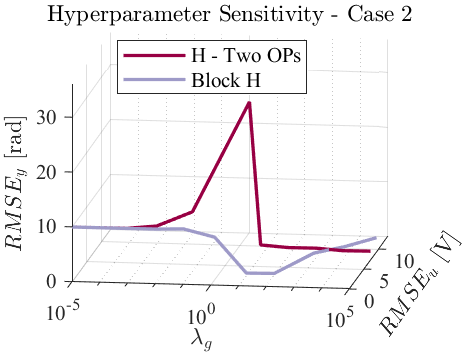}} 
\end{tabular}
    \caption{Sensitivity to $\lambda_g$: performance indicators for the two OPs Hankel \emph{vs} block Hankel for two different \textquotedblleft switching\textquotedblright \ references.}\label{fig:Hyp}
\end{figure}
\begin{table}[!tb]
    \centering
    \caption{Sensitivity to (high) $\lambda_g$: $\text{RMSE}_u$~[V] and $\text{RMSE}_y$~[rad] for the two OPs Hankel \emph{vs} block Hankel for two different \textquotedblleft switching\textquotedblright \ references.}
    \label{tab:performance_lambda}
        \begin{tabular}{lccccc}
    \multirow{3}{*}{} & \multirow{3}{*}{$\lambda_g$} & \multicolumn{2}{c}{Case 1}  & \multicolumn{2}{c}{Case 2}\\ 
    \cline{2-6}
        & & $\text{RMSE}_u$  & $\text{RMSE}_y$ & $\text{RMSE}_u$ & $\text{RMSE}_y$\\
    \hline 
     \multirow{3}{*}{\begin{tabular}[c]{@{}l@{}}Two OPs\end{tabular}} & $10^1$ & 8.41 & 13.31 & 7.55 & 2.00 \\
      & $\mathbf{10^3}$ & 8.67 & \textbf{0.29} & 7.59 & 1.65 \\
      & $\mathbf{10^5}$ & 9.16 & 0.37 & 7.23 & \textbf{1.64}\\
     \hline 
     \multirow{3}{*}{Block} & $\mathbf{10^1}$ & 1.49 & \textbf{0.18} & 2.70 & \textbf{0.36} \\
      & $10^3$ & 1.73 & 0.27 & 6.84 & 1.30 \\
      & $10^5$ & 7.23 & 9.95 & 9.20 & 2.60 \\
     \hline
    \end{tabular} 
\end{table}
We now analyze the impact that the choice of $\lambda_g$ has on closed-loop performance by focusing on the two data structures that lead to a feasible Lasso-DeePC scheme, by considering the following indicators:
\begin{subequations}\label{eq:RMSEs}
\begin{align}
    &\mathrm{RMSE}_{u}=\sqrt{\frac{1}{T_{\mathrm{sim}}}\sum_{t=0}^{T_{\mathrm{sim}}-1}(u_t-u_t^{r})^{2}}~~~\mathrm{[V]},\\
    &\mathrm{RMSE}_{y}=\sqrt{\frac{1}{T_{\mathrm{sim}}}\sum_{t=0}^{T_{\mathrm{sim}}-1}(y_t-y_t^{r})^{2}}~~~\mathrm{[rad]},
\end{align}
\end{subequations}
where $T_{\mathrm{sim}}$ is the length (in steps) of our closed-loop simulations. As shown in \figurename{~\ref{fig:Hyp}}, for small values of $\lambda_g$ the two considered data structures lead to performance that is comparable in terms of both input and output reference tracking, irrespectively of whether the output reference switches from $0.17$~[rad] to either $1.57$~[rad] (case 1) or $2.97$~[rad] (case 2). The differences in performance achieved with the two data structures become nonetheless considerable when $\lambda_g>10^{-1}$, with the block Hankel generally resulting in lower RMSEs with respect to the Two OPs Hankel and the latter leading to better performance only for high values of $\lambda_g$ (see \tablename{~\ref{tab:performance_lambda}}).

\subsection{Performance \emph{vs} explainability}
\begin{figure}[!tb]
\centering
\begin{tabular}{cc}
    \includegraphics[scale=0.5]{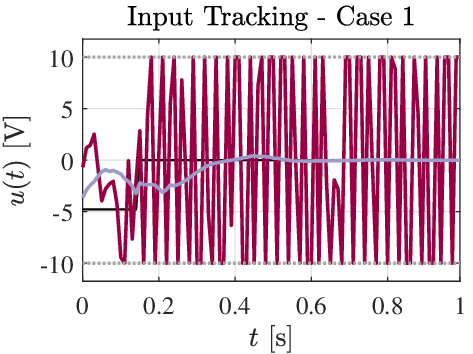} &   \includegraphics[scale=0.5]{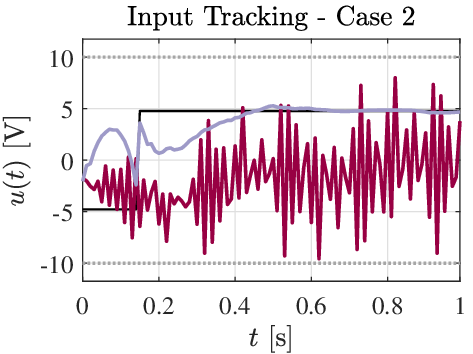}\\
    \includegraphics[scale=0.5]{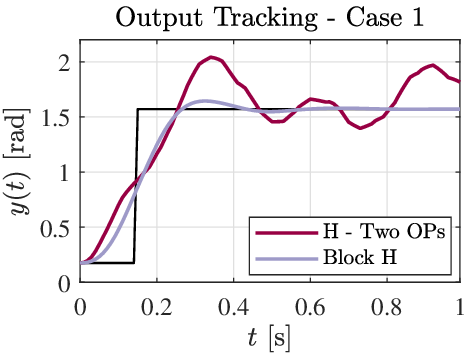} & \includegraphics[scale=0.5]{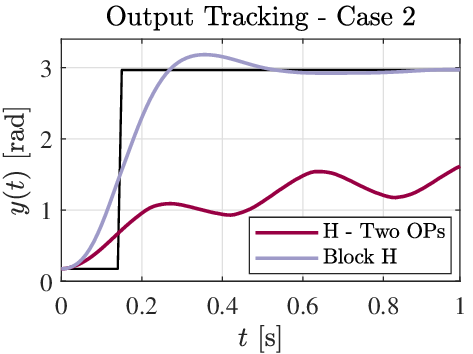}
\end{tabular}
    \caption{Tracking performance of Lasso-based DeePC: two OPs Hankel \emph{vs} block Hankel for two different \textquotedblleft switching\textquotedblright \ references.}
    \label{fig:tracking}
\end{figure}
\begin{figure}[!tb]
\centering
\begin{tabular}{cc}
   \includegraphics[scale=.5]{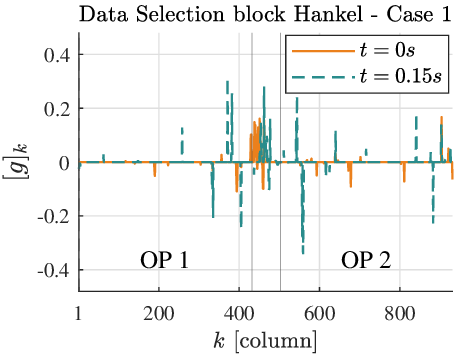}  &  \includegraphics[scale=.5]{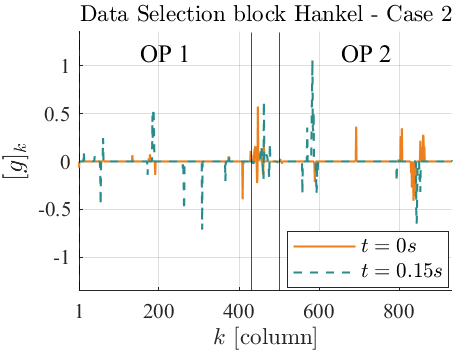}\\
\includegraphics[scale=.5]{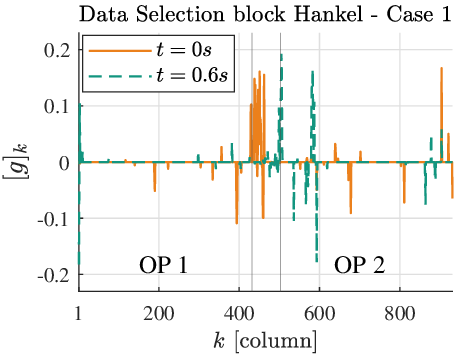} & \includegraphics[scale=.5]{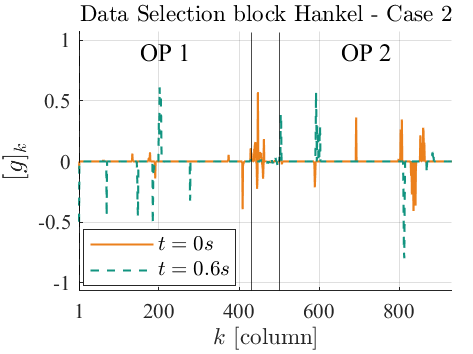}
\end{tabular}
\caption{Values of the selector's component at different instants of the simulation horizon for two different \textquotedblleft switching\textquotedblright \ references.}\label{fig:selectors}
\end{figure}
We now focus on the values of $\lambda_g$ resulting in the least RMSE$_y$ for the block and two OPs Hankels (see the bold values of $\lambda_g$ in \tablename{~\ref{tab:performance_lambda}}). As shown in \figurename{~\ref{fig:tracking}} and expected based on the performance indexes already discussed, only the Lasso-DeePC scheme exploiting the block Hankel matrix allows the unbalanced disk to track the desired reference. Meanwhile, using the Two OPs Hankel results in considerable oscillations of the input in both cases, leading to poor tracking performance when the reference switches from $0.17$ to $1.57$~[rad] (Case 1) and the inability to reach an angular position of $2.97$~[rad] (in Case 2).

By considering only the block Hankel (as it is the only data structure that allows the system to track the desired reference), we then analyze whether performance is paired with an explainable selector (according to Definition~\ref{eq:interpretability_loc_def}). As shown in \figurename{~\ref{fig:selectors}}, the selector picks data collected around both the first and the second operating point irrespective of the considered time instants and, thus, the OP around which we aim the system to operate in closed-loop. This result validates the conclusions drawn in Section~\ref{sec:interpretability}, highlighting that Lasso-DeePC with tailored data structures leads to a control law that (in our tests) allows the unbalanced disk to achieve the desired tracking objectives while being unexplainable. Indeed, Lasso-DeePC does not use only the data characterizing the local behavior of the system around a specific OP, even when tracking such an operating point as a reference. 
\section{Conclusions}\label{sec:conclusions}


In this paper, we analyzed the explainability of data-driven predictive control of nonlinear systems by employing  DeePC with Lasso regularization. Specifically, we derived the explicit solution to 
a proxy of the Lasso-DeePC problem and proposed different strategies for structuring the data matrices used in the DeePC predictor to enhance explainability. Our numerical study on a benchmark nonlinear system highlights the benefits of carefully designed experiments to obtain 
explainable data structures, while also revealing that, although Lasso-DeePC enables numerical tracking, \textit{it can result in an unexplainable solution}.

Future research will focus on incorporating structured regularization in DeePC to explicitly account for groupings in the selector, extending the explainability analysis to other data-driven control schemes for nonlinear systems, and exploring \textit{explainable-by-design} predictive control approaches tailored to a specific class of nonlinear systems like PieceWise Affine (PWA) systems.

\bibliography{V1_main_VB.bib}
\end{document}